\documentclass[pra,twocolumn,showpacs,preprintnumbers,amsmath,amssymb,superscriptaddress]{revtex4}

\usepackage{bm}
\usepackage{graphicx}
\usepackage{amsbsy}
\usepackage{amsmath}
\usepackage{amsfonts}
\usepackage{amsthm}

\begin{document}

\theoremstyle{plain}
\newtheorem{theorem}{Theorem}
\newtheorem{lemma}[theorem]{Lemma}
\newtheorem{corollary}[theorem]{Corollary}
\newtheorem{conjecture}[theorem]{Conjecture}
\newtheorem{proposition}[theorem]{Proposition}

\theoremstyle{definition}
\newtheorem{definition}{Definition}

\theoremstyle{remark}
\newtheorem*{remark}{Remark}
\newtheorem{example}{Example}

\title{Entanglement of Subspaces and Error Correcting Codes}   
\author{Gilad Gour}\email{gour@math.ucalgary.ca}
\affiliation{Institute for Quantum Information Science and 
Department of Mathematics and Statistics,
University of Calgary, 2500 University Drive NW,
Calgary, Alberta, Canada T2N 1N4} 
\author{Nolan R. Wallach}\email{nwallach@ucsd.edu}
\affiliation{Department of Mathematics, University of California/San Diego, 
        La Jolla, California 92093-0112}

\date{\today}

\begin{abstract} 
We introduce the notion of \emph{entanglement of subspaces} as a measure that quantify the entanglement of
bipartite states in a randomly selected subspace. We discuss its properties and in particular we show
that for maximally entangled subspaces it is additive. Furthermore, we show that maximally entangled subspaces
can play an important role in the study of quantum error correction codes.
We discuss both degenerate and non-degenerate codes and show that the subspace spanned
by the logical codewords of a non-degenerate code is a $k$-totally (maximally) entangled
subspace. As for non-degenerate codes, we provide a mathematical definition in terms of subspaces
and, as an example, we analyze Shor's nine qubits code in terms of 22 mutually orthogonal subspaces.
\end{abstract}  


\maketitle

\section{Introduction and definitions}

Bipartite entanglement has been recognized as a crucial resource for quantum information processing tasks
such as teleportation~\cite{Ben93} and super dense coding~\cite{Ben92}. 
As a result, in the last years there has been an 
enormous effort to understand and study the characterization, manipulation and quantification of bipartite entanglement~\cite{Ple07}.  
Yet, despite a great deal of progress that was achieved, the theory on mixed bipartite entanglement is 
incomplete and a few central important questions such as the additivity of the entanglement of formation~\cite{Sho04} remained open.  
Perhaps the richness and complexity of mixed bipartite entanglement can be found in the fact that a finite set of 
measures of entanglement is insufficient to completely quantify it~\cite{Gou05}. 
In this paper we shed some light on mixed bipartite entanglement with the introduction of a new kind of measure of entanglement which we call entanglement of subspaces (EoS). We will see that EoS can play an important role in the study of quantum error correcting codes (QECC).

It has been shown recently~\cite{Hay04,Hay06} that geometry of high-dimensional vector spaces can be counterintuitive especially when subspaces with very unique properties are more common than one intuitively expects. That is, roughly 
speaking, if a high dimensional subspace is selected randomly it is quite likely to have strange properties. For example,
in~\cite{Hay06} it has been demonstrated that a randomly chosen subspace of a bipartite quantum system
will likely contain nothing but nearly maximally entangled states even if the dimension of the subspace is almost of the same order
as the dimension of the original system. This kind of result has implications, in particular, to super-dense coding~\cite{Abe06} 
and for quantum communication in general (see also~\cite{Sho02} for other implications of randomly selected subspaces). 
The quantification of the entanglement of such subspaces is therefore very important and we start with its definition.

\begin{definition}
Let $\mathcal{H}^\text{A}$ and $\mathcal{H}^\text{B}$ be finite dimensional Hilbert spaces and
let $W^\text{AB}$ be a subspace of $\mathcal{H}^\text{A}\otimes\mathcal{H}^\text{B}$.
The entanglement of $W^\text{AB}$ is defined as:
\begin{equation}
\mathcal{E}\left(W^\text{AB}\right)\equiv\min_{\psi^\text{AB}\in W^\text{AB}}
\Big\{E\left(\psi^\text{AB}\right)\;:\;\;\|\psi^\text{AB}\|=1\;\Big\}\;,
\label{def1}
\end{equation}
where $E\left(\psi^\text{AB}\right)$ is the entropy of entanglement of $\psi^\text{AB}$. 
\end{definition}
Note that if the subspace $W^\text{AB}$ contains a product state then $\mathcal{E}(W^\text{AB})=0$.
On the other hand, if, for example, $W^\text{AB}$ is orthogonal to a subspace spaned by an unextendible
product basis (UPB)~\cite{Ben99,DiV03} then $\mathcal{E}(W^\text{AB})>0$.\\
\\
\textbf{Claim}: Let $d_A=\dim\mathcal{H}^\text{A}$ and $d_B=\dim\mathcal{H}^\text{B}$.
If $\mathcal{E}(W^\text{AB})>0$ then
\begin{equation}
\dim W^\text{AB}\leq (d_\text{A}-1)(d_\text{B}-1).
\label{dim}
\end{equation}
This claim follows from~\cite{Wal02} and also related to the fact that the number of (bipartite) states
in a UPB is at least $d_\text{A}+d_\text{B}-1$~\cite{Ben99}.  Note that for two qubits (i.e. $d_A=d_B=2$)
$\mathcal{E}(W^\text{AB})$ can be greater than zero only for one dimensional subspaces. 

We can use Eq.~(\ref{def1}) to define another measure of entanglement on bipartite mixed
states. 
\begin{definition}
Let $\rho\in\mathcal{B}\left(\mathcal{H}^\text{A}\otimes\mathcal{H}^\text{B}\right)$ be
a bipartite mixed state and let $S_{\rho}^\text{AB}$ be the support subspace of $\rho$.
Then, the entanglement of the support of $\rho$ is defined as
$$
E_\text{Support}(\rho)\equiv\mathcal{E}(S_{\rho}^\text{AB})\;. 
$$
\end{definition}
It can be easily seen that this measure is not continuous
and therefore can not be considered as a proper measure of entanglement. Nevertheless,
this measure can serve as a mathematical tool to find lower bounds for other measures of
entanglement that are more difficult to calculate especially in higher dimensions.  
For example, the entanglement of the support
of $\rho$ provides a lower bound for the entanglement of formation. It can be shown that in lower
dimensions the bound is generally not tight. For example, for two qubits in a \emph{mixed} state
$\rho$, the entanglement of the support $E_\text{Support}(\rho)=0$ (see Eq.~(\ref{dim})). On the other
hand, in higher dimensions the bound can be very tight~\cite{Hay04,Hay06}.  

\section{Entanglement of Subspaces}

In this section we study some of the properties of EoS with a focus
on additivity properties. The EoS provides a lower bound on the entanglement of formation
and our interest in its additivity properties is due to one of the most important
unresolved questions in quantum information, namely
the additivity conjecture for the entanglement of formation.
In particular, the additivity question of EoS
is identical to the additivity conjecture of quantum channel output entropy~\cite{AG07}
that has been shown to be equivalent to the additivity conjecture of entanglement of formation~\cite{Sho04}. 
Thus, additivity properties of EoS can shed some light on this topic.

\subsection{Additivity properties of the entanglement of subspaces}

Here we consider the additivity properties of EoS. 
We start by showing that if $U^\text{AB}$ and $V^{\text{A}'\text{B}'}$ are two subspaces
such that $\mathcal{E}(U^\text{AB})>0$ and/or $\mathcal{E}(V^{\text{A}'\text{B}'})>0$ then
$\mathcal{E}(U^\text{AB}\otimes V^{\text{A}'\text{B}'})>0$.
 
Consider $W=\mathbb{C}^{n}\otimes\mathbb{C}^{m}$. 
Let $e_{j}$, $j=1,...,n$ be the standard basis of
$\mathbb{C}^{n}$. We will also use the notation $f_{j}$ for the standard basis
of $\mathbb{C}^{m}.$ An element of a tensor product of two vector spaces, $A$
and $B$ will be called a product if it is of the form $a\otimes b$ 
with $a\in A$ and $b\in B$.

\begin{proposition}
Let $u_{1},...,u_{d},v_{1},...,v_{d}\in W$ be such that if 
$x=\sum_{i} b_{i}v_{i}$ is a product then $x=0$. 
If $z=\sum_{i} u_{i}\otimes v_{i}$ 
is a product in 
$\left(\mathbb{C}^{n} \otimes \mathbb{C}^{n}\right) \otimes \left(\mathbb{C}^{m} \otimes \mathbb{C}^{m}\right)$ then $z=0$.
\end{proposition}

\begin{proof}
We write $u_{i}=\sum_{j=1}^{n}e_{j} \otimes u_{ij}$ 
and $v_{j}=\sum_{j=1}^{n}e_{j} \otimes v_{ij}$. 
Assume that ${\sum_{i}}u_{i}{\otimes}v_{i}$ 
is a product in 
$(\mathbb{C}^{n}{\otimes}\mathbb{C}^{n}){\otimes}(\mathbb{C}^{m}{\otimes}\mathbb{C}^{m})$. 
This means that there exists 
$z\in\mathbb{C}^{n}{\otimes}\mathbb{C}^{n}$ and $w\in\mathbb{C}^{m}{\otimes}\mathbb{C}^{m}$ such that
\[
{\sum_{i,k,l}}\left(  e_{k}{\otimes}e_{l}\right)
{\otimes}(u_{ik}{\otimes}v_{il})=z{\otimes}w.
\]
If we write out 
$z={\sum_{k,l}}z_{kl}e_{k}{\otimes}e_{l}$ 
with $z_{kl}\in\mathbb{C}$ then we must have
\[
{\sum_{i}}u_{ik}{\otimes}v_{il}=z_{kl}w
\]
for all $k,l$. \ We now write 
$u_{ik}={\sum_{m}}u_{ik}^{m}f_{m}$ and $w={\sum_{m}}f_{m}{\otimes}w_{m}$. 
The displayed formula now implies that ($k.l,m$ fixed)
\[
{\sum_{i}}
u_{ik}^{m}v_{il}=z_{kl}w_{m}.
\]
This implies that (with $k$ and $m$ fixed) we have
\[
{\sum_{il}}u_{ik}^{m}e_{l}
{\otimes}v_{il}=({\sum_{l}}z_{kl}e_{l})
{\otimes}w_{m}.
\]
Hence 
${\sum_{i}}u_{ik}^{m}v_{i}$ is a product. Our assumption implies that it must be $0$.
Hence
\[
0={\sum_{i,k,m}}u_{ik}^{m}e_{k}
{\otimes}f_{m}{\otimes}
v_{i}={\sum_{i}}u_{i}{\otimes}v_{i}.
\]
As was to be proved.
\end{proof}

Note that the proposition above states that if none of the decompositions of a bipartite mixed
state, $\rho$, contain a product state, then also none of the decompositions of 
$\rho\otimes\sigma$ ($\sigma$ is a bipartite mixed state) contain a product state.
This property is related to the additivity conjecture~\cite{Sho04} for the entanglement of formation
(and other measures) and one of the main questions that we will consider here is 
wether the EoS is additive. That is, does 
$$
\mathcal{E}(U^\text{AB}\otimes V^{\text{A}'\text{B}'})=\mathcal{E}(U^\text{AB})
+\mathcal{E}(V^{\text{A}'\text{B}'})\;?
$$
Clearly, if the EoS were additive then the proposition above would have been
a trivial consequence of that. However, we were not able to prove the additivity of EoS (in general) 
although for some special cases it has been tested numerically in~\cite{HarXX}  and no counter example 
has been found. The proposition below provides a lower bound.
\begin{proposition}
Let $N=\min\{\dim U^\text{AB},\; \dim V^{\text{A}'\text{B}'}\}$. Then
\begin{equation}
\mathcal{E}(U^\text{AB})+\mathcal{E}(V^{\text{A}'\text{B}'})-\log N
\leq\mathcal{E}\left(U^\text{AB}\otimes V^{\text{A}'\text{B}'}\right)\;.
\end{equation} 
\end{proposition} 
The equation above provides a lower bound whereas the upper bound 
$\mathcal{E}\left(U^\text{AB}\otimes V^{\text{A}'\text{B}'}\right)\leq\mathcal{E}(U^\text{AB})+\mathcal{E}(V^{\text{A}'\text{B}'})$ 
follows directly from the definition of EoS.
Thus, for $N=1$ the EoS is additive. 
Note also that even if $N$ is small (e.g. $N=2$), $\mathcal{E}(U^\text{AB})$ and 
$\mathcal{E}(V^{\text{A}'\text{B}'})$ can be arbitrarily large (i.e. depending on $d_A$ and $d_B$ but not on $N$).

\begin{proof}
Let $\chi$ be a normalized vector in $U^\text{AB}\otimes V^{\text{A}'\text{B}'}$. We can write $\chi$ in its
Schmidt decomposition as follows:
$$
\chi=\sum_{i}\sqrt{p_{i}}u_{i}^\text{AB}\otimes v_{i}^{\text{A}'\text{B}'}\;,
$$ 
where $\sum_{i}p_i=1$ ($p_i \geq 0$) and the $u_{i}^\text{AB}$'s ($v_{i}^{\text{A}'\text{B}'}$'s)
are orthonormal. Now, from the strong subadditivity of the von-Neumann entropy we have
$$
S(\rho_{\text{A}'})+S(\rho_\text{B})\leq S(\rho_\text{AB})+S(\rho_{\text{AA}'})\;,
$$
where $\rho_\text{A}\equiv {\rm Tr}_{\text{A}'\text{BB}'}\chi\otimes \chi^{*}$,
$\rho_\text{B}\equiv {\rm Tr}_{\text{AA}'\text{B}'}\chi\otimes \chi^{*}$, etc.
Now, note that $S(\rho_{\text{AA}'})=E(\chi)$ and $S(\rho_\text{AB})=H(\{p_i\})\leq\log N$,
where $H(\{p_i\})$ is the Shanon entropy. Furthermore, note that
$$
\rho_{\text{A}'}=\sum_{i}p_i \omega _i\;\;{\rm and}\;\;\rho_\text{B}=\sum_{i}p_i \sigma _i
$$
where
$$
\omega_{i}\equiv {\rm Tr}_{\text{B}'}v_{i}^{\text{A}'\text{B}'}\otimes {v_{i}^{\text{A}'\text{B}'}}^{*}\;\;{\rm and}\;\;
\sigma_{i}\equiv {\rm Tr}_\text{A}u_{i}^\text{AB}\otimes {u_{i}^\text{AB}}^{*}\;.
$$
Hence, since the von-Neumann entropy is concave we have
$$
S(\rho_{\text{A}'})\geq\sum_{i}p_i S(\omega_i)=\sum_{i}p_i E(v_{i}^{\text{A}'\text{B}'})
\geq \mathcal{E}\left(V^{\text{A}'\text{B}'}\right)\;.
$$
and similarly $S(\rho_\text{B})\geq\mathcal{E}\left(U^\text{AB}\right)$. Combining all this
we get
$$
\mathcal{E}\left(V^{\text{A}'\text{B}'}\right)+\mathcal{E}\left(U^\text{AB}\right)\leq\log N+E(\chi)\;,
$$
for all $\chi\in U^\text{AB}\otimes V^{\text{A}'\text{B}'}$. This complete the proof.
\end{proof}

\subsection{Maximally entangled subspaces}

As we have seen above, if $N=1$ then the EoS is clearly additive. As we will see
in the next subsection, it is also additive for maximally entangled subspaces: 
\begin{definition}
Let $W$ be a subspace of $\mathcal{H}^\text{A}\otimes\mathcal{H}^\text{B}$
and let $d_\text{A}=\dim \mathcal{H}^\text{A}$ and $d_\text{B}=\dim \mathcal{H}^\text{B}$.
$W$ is said to be a maximally entangled subspace in
$\mathcal{H}^\text{A}\otimes\mathcal{H}^\text{B}$ if
\begin{equation}
\mathcal{E}(W)=\log m\;,
\end{equation}
where $m\equiv\min \{d_\text{A},\;d_\text{B}\}$.
\end{definition}
The term maximally entangled subspace have been used in~\cite{Hay04,Hay06} for
a subspace $W$ with $\mathcal{E}(W)$ slightly less than $\log m$. In this paper, 
we will call such subspaces nearly maximally entangled to distinguish
from (exactly) maximally entangled subspaces as defined above.

In~\cite{Pag93}  it has been shown that the average entanglement of a pure state 
$\phi\in \mathcal{H}^\text{A}\otimes\mathcal{H}^\text{B}$ which is chosen randomly 
according to the unitarily invariant measure satisfies
$$
\langle E(\phi)\rangle \geq \log_{2}d_\text{A}-\frac{d_\text{A}}{2\ln 2 d_\text{B}}
$$
where without loss of generality $d_\text{A}\geq d_\text{B}$. Later on, in~\cite{Hay04,Hay06}
this result has been extended to subspaces and in particular it has been shown, 
somewhat surprisingly, that a randomly chosen subspace of bipartite quantum 
system will likely be a nearly maximally entangled subspace. Thus, as nearly 
maximally entangled subspaces are quite common it is important to understand 
their structure. As a first step in this direction, in the following we study the 
structure of (exactly) maximally entangled subspaces. 

Let $\phi$ be a state in $\mathcal{H}^\text{A}{\otimes}\mathcal{H}^\text{B}$.
If $e_{1},...,e_{m}$ is an orthonormal basis of $\mathcal{H}^\text{B}$ we may
write
$$
\phi=\sum_{i=1}^{d_{B}}\phi_{i}{\otimes}e_{i}\;.
$$
We define a $d_\text{B}\times d_\text{B}$ Hermitian matrix 
$B=\left[  \left\langle\phi_{i}|\phi_{j}\right\rangle \right]$
(i.e. $B$ is the reduced density matrix). 
Let $\lambda_{1},...,\lambda_{d_\text{B}} $ be the set of 
eigenvalues of $B$ counting multiplicity. Then the
entanglement of $\phi$  is
$$
E(\phi)=-{\sum_{i}}\lambda_{i}\log(\lambda_{i})\;.
$$
It is easy to show that $E(\phi)\leq\log m$
and equality is attained if and only if $B=\frac{1}{m}P$ with $P$ a projection
matrix onto a $d$ dimensional subspace of $\mathbb{C}^{d_\text{B}}$. Clearly
this definition of entropy is independent of the choice of basis and could
also be given using an orthonormal basis of $\mathcal{H}^\text{A}$ and analyzing the
corresponding $d_\text{A}$ coefficients in $\mathcal{H}^\text{B}$. Under the
condition of equality $\phi$ is maximally entangled, and this in particular
implies that if $d_\text{A}\geq d_\text{B}$ then
$$
\left\langle \phi_{i}|\phi_{j}\right\rangle =\frac{1}{d_\text{B}}\delta_{ij}.
$$

\begin{proposition}
Assume that $d_\text{A}\geq d_\text{B}$ and set $m=d_{B}$. Let $U^\text{AB}$ be a maximally
entangled subspace in $\mathcal{H}^\text{A}{\otimes}\mathcal{H}^\text{B}$ of dimension $d$. 
If $e_{1},...,e_{m}$ is an orthonormal basis of $\mathcal{H}^\text{B}$ then there
exist $U_{1},...,U_{m}$ subspaces of $\mathcal{H}^\text{A}$ such that 
$\left\langle U_{i}|U_{j}\right\rangle =0$ if $i\neq j$, $\dim U_{j}=d>0$ for all
$j=1,...,m$ and unitary operators $T_{i}:\mathbb{C}^{d}\rightarrow U_{i}$
$i=1,...,m$ such that
$$
U^\text{AB}=\{{\sum_{i}}T_{i}w{\otimes}e_{i}|w\in\mathbb{C}^{d}\}.
$$
Conversely, if $U_{1},...,U_{m}$ are mutually orthogonal subspaces of $\mathcal{H}^\text{A}$
such that $\dim U_{j}=d>0$ for all $j=1,...,m$ and we have unitary operators
$T_{i}:\mathbb{C}^{d}\rightarrow U_{i}$ $i=1,...,m$ such that
$$
U^\text{AB}=\{\sum_{i} T_{i}w{\otimes}e_{i}|w\in\mathbb{C}^{d}\}\;,
$$
then $U^\text{AB}$ is maximally entangled.
\label{structure}
\end{proposition}

\begin{proof} 
Let $\psi_{1},...,\psi_{d}$ be an orthonomal basis of $U^\text{AB}$. Then we can write
$$
\psi_{j}={\sum_{i}}\psi_{ij}{\otimes}e_{i}
$$
with $\left\langle \psi_{ij}|\psi_{kj}\right\rangle =\frac{1}{m}\delta_{ik} $.
The condition on $U^\text{AB}$ is that if $a\in\mathbb{C}^{d}$ is a unit vector then 
$\sum a_{j}\psi_{j}$ is maximally entangled in 
$\mathcal{H}^\text{A}{\otimes}\mathcal{H}^\text{B}$. This implies that
$$
\left\langle{\sum_{j=1}^{d}}a_{j}\psi_{lj}\Big|{\sum_{j=1}^{d}}
a_{j}\psi_{kj}\right\rangle =\frac{1}{m}\delta_{l,k}\;.
$$
Fix $l\neq k$ and let $p\neq q\leq d$ be two integers. Let 
$a=(a_{1},...,a_{d})$ with $a_{j}=0$
for $j\neq p$ or $j\neq q$. Set $a_{p}=b,a_{q}=c$ and $|b|^{2}+|c|^{2}=1$. Then we
have
$$
\left\langle b\psi_{lp}+c\psi_{lq}|b\psi_{kp}+c\psi_{kq}\right\rangle =0.
$$
On the other hand we have
$$
\left\langle b\psi_{lp}+c\psi_{lq}|b\psi_{kp}+c\psi_{kq}\right\rangle
=\overline{b}c\left\langle \psi_{lp}|\psi_{kq}\right\rangle +\overline
{c}b\left\langle \psi_{lq}|\psi_{kp}\right\rangle
$$
Set $z=\overline{b}c$. We look at two cases: first $z=\frac{1}{2}$
($b=c=\frac{1}{\sqrt{2}})$ and second $z=\frac{i}{\sqrt{2}}$ ($b=\frac{1}{\sqrt{2}
},c=\frac{i}{\sqrt{2}}$). Thus we have
$$
\left\langle \psi_{lp}|\psi_{kq}\right\rangle +\left\langle \psi_{lq}
|\psi_{kp}\right\rangle =0
$$
for the first case and
$$
\left\langle \psi_{lp}|\psi_{kq}\right\rangle -\left\langle \psi_{lq}
|\psi_{kp}\right\rangle =0
$$
for the second.
Hence, $\left\langle \psi_{lp}|\psi_{kq}\right\rangle =\left\langle \psi
_{lq}|\psi_{kp}\right\rangle =0$. We set $U_{l}={\rm Span}\{\psi_{lp}|p=1,...,d\} $.
Then $\left\langle U_{l}|U_{k}\right\rangle =0$ if $l\neq k$. We now consider
what happens when $l=k$. We first note that taking $a_{p}=1$ and all other
entries equal to $0$ we have $\left\langle \psi_{lp}|\psi_{lp}\right\rangle
=\frac{1}{m}$. Now using $b,c$ as above for $p\neq q$ we have
$$
\left\langle \psi_{lp}|\psi_{lp}\right\rangle +\left\langle \psi_{lp}
|\psi_{lq}\right\rangle +\left\langle \psi_{lq}|\psi_{lp}\right\rangle
+\left\langle \psi_{kp}|\psi_{kp}\right\rangle =\frac{2}{m}
$$
and
$$
\left\langle \psi_{lp}|\psi_{lp}\right\rangle +i\left\langle \psi_{lp}
|\psi_{lq}\right\rangle -i\left\langle \psi_{lq}|\psi_{lp}\right\rangle
+\left\langle \psi_{kp}|\psi_{kp}\right\rangle =\frac{2}{m}.
$$
Hence as above we find that $\left\langle \psi_{lp}|\psi_{lq}\right\rangle =0$
if $p\neq q$. Thus $\sqrt{m}\psi_{l1},...,\sqrt{m}\psi_{ld}$ is an orthonormal
basis of $U_{l}$. This implies that the spaces $U_{1},...,U_{m}$ have the
desired properties. Let $u_{1},...,u_{d}$, be the standard orthonormal basis
of $\mathbb{C}^{d}$ and define $T_{i}u_{j}=\sqrt{m}\psi_{ij}$. With this
notation in place $U^\text{AB}$ has the desired form.  
The converse is proved by the obvious calculation.
\end{proof}

\begin{corollary}\label{cordim}
If $U^\text{AB}$ is a maximally entangled subspace in 
$\mathcal{H}^\text{A}{\otimes}\mathcal{H}^\text{B}$ ($d_\text{A}\geq d_\text{B}$), then
$$
\dim U^\text{AB} \leq \left\lfloor\frac{d_\text{A}}{d_\text{B}}\right\rfloor\;.
$$
Furthermore, there always exists a maximally entangled subspace of dimension
$\lfloor d_\text{A}/d_\text{B}\rfloor$.
\end{corollary}
\begin{proof}
Assume that $\dim\mathcal{H}^\text{A}=d_\text{A}\geq\dim\mathcal{H}^\text{B}=d_\text{B}$.
According to the first part of Proposition 4, if $U^\text{AB}$ is a maximally
entangled subspace of dimension $d$ then $d\times d_\text{B}\leq d_\text{A}$. On the
other hand, if $d\leq\left\lfloor d_\text{A}/d_\text{B}\right\rfloor $ then
the second half of the statement implies that there is a maximally entangled
subspace of dimension $d$.
\end{proof}

In the following we find necessary and sufficient conditions for a subspace to be maximally 
entangled. In section III we use this to show that maximally entangled subspaces can play an important
role in the study of error correcting codes. 
As above we consider the space $\mathcal{H}^\text{A}{\otimes}\mathcal{H}^\text{B}$
with $\dim\mathcal{H}^\text{A}=d_\text{A}\geq\dim\mathcal{H}^\text{B}=d_\text{B}$ and
a maximally entangled subspace $U^\text{AB}\subset\mathcal{H}^\text{A}{\otimes}\mathcal{H}^\text{B}$.
We will also consider $\text{End}(\mathcal{H}^\text{B})$ to be a Hilbert space with inner product 
$\langle X|Y\rangle=\text{Tr}(X^{\dag}Y)$ for any two operators $X$ and $Y$ in  
$\text{End}(\mathcal{H}^\text{B})$.

\begin{proposition}
\label{isometry}
Let $U^\text{AB}\subset\mathcal{H}^\text{A}{\otimes}\mathcal{H}^\text{B}$ be a subspace and 
$d_\text{A}\geq d_\text{B}$. Then,
$U^\text{AB}$ is maximally entangled if and only if the map 
$\text{End}(H^\text{B})\otimes U^\text{AB}\rightarrow \mathcal{H}^\text{A}{\otimes}\mathcal{H}^\text{B}$
given by $X\otimes u\mapsto\sqrt{d_\text{B}}(I\otimes X)u$ is an isometry onto its image. 
\end{proposition}

\begin{proof}
Let $d=\dim U^\text{AB}$ and let the notation be as in Proposition~\ref{structure}. Thus, 
if $U^\text{AB}$ is maximally entangled and if
$e_1,...,e_{d_\text{B}}$ is an orthonormal basis of $\mathcal{H}^\text{B}$ then an element
of $U^\text{AB}$ is of the form
$$
T(w)=\sum_{i=1}^{d_\text{B}}T_i(w)\otimes e_i\;,
$$
with $T_i$ a unitary operator from $\mathbb{C}^d$ onto a subspace $U_i$ of $\mathcal{H}^\text{A}$
and $U_i$ and $U_j$ are orthogonal for all $i\neq j$. We now calculate
$$
\langle (I\otimes X)T(w)|(I\otimes Y)T(z)\rangle=\sum_{i,j}\langle T_i(w)|T_j(z)\rangle\;\langle Xe_i|Ye_j\rangle\;.
$$
Now, since $\langle T_i(w)|T_j(z)\rangle=\delta_{ij}\langle w|z\rangle$ (see Proposition~\ref{structure})
we have:
\begin{align}
\langle (I\otimes X)T(w)|(I\otimes Y)T(z)\rangle & =\sum_{i,j}{\delta_{ij}\langle w|z\rangle}\;\langle Xe_i|Ye_j\rangle\nonumber\\
& ={\langle w|z\rangle\text{Tr}(X^\dag Y)}\;.\nonumber
\end{align}
That is, we proved that if $U^\text{AB}$ is maximally entangled then the map is an isometry.  For the converse we note
that we have an isometry of $\mathbb{C}^d$ onto $U^\text{AB}$ given by
$$
T(w)=\sum_{i=1}^{d_\text{B}} T_i(w)\otimes e_i\;.
$$
Now, if the map defined in the proposition is an isometry then
$$
\langle (I\otimes X)T(w)|(I\otimes Y)T(z)\rangle = {\langle w|z\rangle\text{Tr}(X^\dag Y)}\;.
$$
That is,
$$
\sum_{i,j}\langle T_i(w)|T_j(z)\rangle\;\langle Xe_i|Ye_j\rangle
={\langle w|z\rangle\text{Tr}(X^\dag Y)}\;,
$$
for all $X,Y\in\text{End}(\mathcal{H}^\text{B})$. Hence, we must have 
$\langle T_i(w)|T_j(z)\rangle=\delta_{ij}\langle w|z\rangle /$
and from Proposition~\ref{structure} the subspace $U^\text{AB}$ is maximally entangled.
\end{proof}

\subsection{Additivity of maximally entangled subspaces}

We now discuss the additivity properties of maximally entangled subspaces. 
\begin{proposition}\label{additive}
Let $U^\text{AB}\subset \mathcal{H}^\text{A}{\otimes}\mathcal{H}^\text{B}$ and 
$V^{\text{A}'\text{B}'}\subset \mathcal{H}^{\text{A}'}{\otimes}\mathcal{H}^{\text{B}'}$ be 
maximally entangled subspaces. Then,
\begin{align}
\mathcal{E}\left(U^\text{AB}\otimes V^{\text{A}'\text{B}'}\right)
&=\mathcal{E}\left(U^\text{AB}\right)+\mathcal{E}\left(V^{\text{A}'\text{B}'}\right)\nonumber\\
&=\log m+\log m'\;,
\end{align} 
where $m\equiv\min\{d_\text{A},\;d_\text{B}\}$ and 
$m'\equiv\min\{d_{\text{A}'},\;d_{\text{B}'}\}$.
\end{proposition}

\begin{remark}
From the above proposition it follows that if $d_{A}\geq d_{B}$ and $d_{\text{A}'}\geq d_{\text{B}'}$ 
or $d_{B}\geq d_{A}$ and $d_{\text{B}'}\geq d_{\text{A}'}$ then 
$U^\text{AB} {\otimes} V^{\text{A}'\text{B}'}$ is maximally entangled in 
$(\mathcal{H}^\text{A}{\otimes}\mathcal{H}^{\text{A}'}){\otimes}(\mathcal{H}^\text{B}{\otimes}\mathcal{H}^{\text{B}'})$.
However, if for example $d_{A} > d_{B}$ and $d_{\text{A}'} < d_{\text{B}'}$  
then $U^\text{AB} {\otimes} V^{\text{A}'\text{B}'}$ 
is NOT maximally entangled in 
$(\mathcal{H}^\text{A}{\otimes}\mathcal{H}^{\text{A}'}){\otimes}(\mathcal{H}^\text{B}{\otimes}\mathcal{H}^{\text{B}'})$
because  $mm'<\min\{d_\text{A}d_{\text{A}'},\;d_\text{B}d_{\text{B}'}\}$.
\end{remark}

\begin{proof}
There are basically two possibilities (up to interchanging factors): the first
is $d_{A}\geq d_{B}$ and $d_{\text{A}'}\geq d_{\text{B}'}$, and the second is
$d_{A}\geq d_{B}$ and $d_{\text{A}'} < d_{\text{B}'}$.

In the first case we have as in the statement of proposition~\ref{structure} 
the subspaces $U_{j}$ and the unitaries $T_{j}:\mathbb{C}^{d}\rightarrow U_{j}$ such that 
$U^\text{AB}=\{\sum_{i} T_{i}w {\otimes}e_{i}|w\in\mathbb{C}^{d}\}$. 
We also have the orthonormal basis $f_{i}$ of
$\mathcal{H}^{\text{B}'}$, the subspaces $V_{j}$ and the unitaries 
$S_{j}:\mathbb{C}^{d^{\prime}}\rightarrow V_{j}$ such that 
$V^{\text{A}'\text{B}'}=\{\sum_{i}S_{i}w'\otimes f_{i}|w'\in\mathbb{C}^{d^{\prime}}\}$. 
Thus, as a subspace of 
$(\mathcal{H}^\text{A}{\otimes}\mathcal{H}^{\text{A}'}){\otimes}(\mathcal{H}^\text{B}{\otimes}\mathcal{H}^{\text{B}'})$, 
$U^\text{AB} {\otimes} V^{\text{A}'\text{B}'}$ is spanned by the elements
$$
{\sum_{i,j}}(T_{i}w{\otimes}S_{j}w^{\prime}){\otimes}(e_{i}{\otimes}f_{j}).
$$
Thus if we identify $\mathbb{C}^{d}{\otimes}\mathbb{C}^{d^{\prime}}$ with 
$\mathbb{C}^{dd^{\prime}}$ then the converse assertion in proposition~\ref{structure} implies that 
$U^\text{AB} {\otimes} V^{\text{A}'\text{B}'}$
is a maximally entangled space. This implies that 
$$
\mathcal{E}\left(U^\text{AB} {\otimes} V^{\text{A}'\text{B}'}\right)
=\log d_\text{B} +\log d_{\text{B}'}=\log m+\log m'\;.
$$

We now consider the second case. For $U^\text{AB}$ we have exactly as above
$U^\text{AB}=\{\sum_{i} T_{i}w {\otimes}e_{i}|w\in\mathbb{C}^{d}\}$. For 
$V^{\text{A}'\text{B}'}$ we denote by $f_j$ an orthonormal basis of 
$\mathcal{H}^{\text{A}'}$ (not of $\mathcal{H}^{\text{B}'}$ as above).
Thus, according to proposition~\ref{structure} we have 
$V^{\text{A}'\text{B}'}=\{\sum_{i}f_{i}\otimes S_{i}w'|w'\in\mathbb{C}^{d^{\prime}}\}$.
As a subspace of 
$(\mathcal{H}^\text{A}{\otimes}\mathcal{H}^{\text{A}'}){\otimes}(\mathcal{H}^\text{B}{\otimes}\mathcal{H}^{\text{B}'})$, 
$U^\text{AB} {\otimes} V^{\text{A}'\text{B}'}$ is spanned by the elements
$$
{\sum_{i,j}}(T_{i}w{\otimes}f_{j}){\otimes}(e_{i}{\otimes}S_{j}w^{\prime}).
$$
We will assume first that $d^{\prime}\leq d$. Let $w_{1}^{\prime
},...,w_{d^{\prime}}^{\prime}$ be an orthonormal basis of 
$\mathbb{C}^{d^{\prime}}$. Thus, if $\phi$ is a state in $U^\text{AB} {\otimes} V^{\text{A}'\text{B}'}$ 
we can write it as
$$
\phi={\sum_{i,j,k}}(T_{i}w_{k}{\otimes}f_{j}){\otimes}(e_{i}{\otimes}S_{j}w_{k}^{\prime})\;,
$$
where $w_k$ are some non-normalized vectors in $\mathbb{C}^d$. Furthermore, 
$$
\left\langle \phi|\phi\right\rangle =d_\text{B}d_{\text{A}'}{\sum_{k}}
\left\Vert u_{k}\right\Vert ^{2}\;.
$$
Hence, if $\phi$ is normalized then
$$
{\sum_{k}}
\left\Vert u_{k}\right\Vert ^{2}=\frac{1}{d_\text{B}d_{\text{A}'}}\;.
$$
Now, since $S_{j}w_{k}^{\prime}$ is an orthonormal  set of vectors for all $j$ and $k$
(see proposition~\ref{structure}), the entanglement of $\phi$ as an element of
$(\mathcal{H}^\text{A}{\otimes}\mathcal{H}^{\text{A}'}){\otimes}(\mathcal{H}^\text{B}{\otimes}\mathcal{H}^{\text{B}'})$
is given by
$$
d_\text{B}d_{\text{A}'}S(B)
$$
where $B=\left[  \left\langle w_{i}|w_{j}\right\rangle \right]  _{1\leq
i,j\leq d^{\prime}}$, and if $\lambda_{1},...,\lambda_{d^{\prime}}$ are the
eigenvalues of $B$ then the von-Neumann entropy of $B$ is
$$
S(B)=-\sum\lambda_{i}\log\lambda_{i}\;.
$$
Now $B$ is the most general $d^{\prime}\times d^{\prime}$ self adjoint,
positive semidefinate matrix with trace $1/d_\text{B}d_{\text{A}'}$. The minimum of
the entropy for such matrices is
$$
\frac{\log(d_{a}d_{B})}{d_{a}d_{B}}.
$$
This proves the proposition for the case $d'\leq d$. If $d<d^{\prime}$ then we can prove the proposition
by using the same argument, this time with $w_{k}$ an orthonormal basis of $\mathbb{C}^{d}$ and  
with $B^{\prime}=\left[  \left\langle w_{i}^{\prime}|w_{j}^{\prime}\right\rangle \right]  _{1\leq i,j\leq d}$. 
This completes the proof.
\end{proof}

The above proposition also shows that the entanglement of formation is additive for
bipartite states with maximally entangled support. 
If $\rho$ is a mixed state in $\mathcal{H}^\text{A}{\otimes}\mathcal{H}^\text{B}$ 
then the entanglement of formation is defined in terms of the convex roof extension:
$$
E_\text{F}(\rho)=\min\sum p_{i}E(\phi_{i})
$$
where the minimum taken over all decompositions
$$
\rho=\sum p_{i}\phi_{i}{\otimes}\phi_{i}^{\ast}
$$
with $\phi_{i}$ a pure bipartite state and $p_{i}>0$ and $\sum p_{i}=1$.

\begin{corollary}
Let $\rho$ and $\sigma$ be mixed states in 
$\mathcal{B}(\mathcal{H}^\text{A}{\otimes}\mathcal{H}^\text{B})$ and 
$\mathcal{B}(\mathcal{H}^{\text{A}'}{\otimes} \mathcal{H}^{\text{B}'})$, respectively. 
If the support subspaces $S_\rho$ and $S_\sigma$ are maximally entangled
then
$$
E_\text{F}(\rho\otimes\sigma)=E_\text{F}(\rho)+E_\text{F}(\sigma)\;.
$$
\end{corollary}
The proof of this corollary follows directly from the fact that for states with maximally
entangled support $E_\text{F}(\rho)=\mathcal{E}(S_\rho)$.
Note that the class of mixed states with maximally entangled support is extremely small
(i.e. of measure zero). In particular, it is a much smaller class than the one found by
Vidal, Dur and Cirac~\cite{Vid02} .

\section{Error Correcting Codes}

\subsection{Definitions}

We consider error correcting codes that are used to encode $l$ qubits in $n\geq l$ qubits in such
a way that they can correct errors on any subset of $k$ or fewer qubits. These codes, which we call
$(n,l,k)$ error correcting codes, 
can be classified into two classes (for example see~\cite{NC}): degenerate and 
non-degenerate (or orthogonal) codes. We start with a general definition of error correcting codes that is 
equivalent to the definition given (for example) in~\cite{NC} , but here we define the codes in terms 
of subspaces. 

\begin{definition}
Let $X\in End(\otimes^{k}\mathbb{C}^{2})$ and 
$0\leq i_{0}<i_{1}<...<i_{k-1} \leq n-1$. The operator
$X_{i_{0}i_{1} \cdots i_{k-1}}$ on $\otimes^{n}\mathbb{C} ^{2}$, that represents
the errors on the k qubits $i_{1},...,i_{k-1}$, is defined by
$X_{i_{0}...i_{k-1}}v=\sigma (X\otimes I)\sigma^{-1}v$, where {\rm (a)} $\sigma\in S_{n}$ (acting on $\{0,1,...,n-1\}$ by
permutations) is defined such that $\sigma(j)=i_{j}$, {\rm (b)} $\sigma$ can
act on $\otimes^{n}\mathbb{C} ^{2}$ by $\sigma(v_{0}\otimes
v_{1}\otimes\cdots\otimes v_{n-1})=v_{\sigma (0)}\otimes
v_{\sigma (1)}\otimes\cdots\otimes v_{\sigma(n-1)}$ and {\rm (c)} 
$\otimes ^{n} \mathbb{C}^{2}$ is viewed as $(\otimes^{k} \mathbb{C}
^{2})\otimes(\otimes^{n-k} \mathbb{C} ^{2})$ (putting together the
$k$ tensor factors that correspond to the $k$ qubits $i_{1},...,i_{k-1}$ and the rest $n-k$ tensor factors). 
An $(n,l,k)$ error correcting code is defined from its following ingredients:

I. An isometry
$T:\otimes^{l}\mathbb{C}^{2}\rightarrow\otimes^{n}\mathbb{C} ^{2}$.

II. Let $V_{0}=T(\otimes^{l}\mathbb{C} ^{2})$. There are
$V_{1},...,V_{d}$ mutually orthogonal subspaces of
$\otimes^{n}\mathbb{C} ^{2}$ that are also orthogonal to $V_{0}$.

III. For each $V_{j}$ there is a unitary isomorphism, $U_j$, of $V_{j}$ onto $V_{0}$
with $U_{0}=I$.

IV. $X_{i_{0}i_{1}\cdots
i_{k-1}}V_{0}\subset\oplus_{j=0}^{d}V_{j}$.

V. Let $P_{j}$ be the ortogonal projection of $\otimes^{n}\mathbb{C}
^{2}$ onto $V_{j}$ then if $v\in V_{0}$ is a unit vector and $P_{j}
(X_{i_{0}i_{1}\cdots i_{k-1}}v)\neq0$ then
\[U_{j}\frac{P_{j}(X_{i_{0}i_{1}\cdots i_{k-1}}v)}{\left\Vert P_{j}
(X_{i_{0}i_{1}\cdots i_{k-1}}v)\right\Vert }\] equals $v$ up to
a phase.
\end{definition}

In the next subsection we study Shor's $(9,1,1)$ error correcting code and show that it 
satisfies this definition. However, before that, let us introduce the notion of 
$k$-\emph{totally entangled} subspaces which will play an important role in our 
discussion of QECC.

\begin{definition}
Let $\mathcal{H}$ be the space of $n$ qubits, $\otimes ^{n}\mathbb{C}^2$. Corresponding 
to any choice of $k$ qubits (tensor factors) we can consider 
$\mathcal{H}=\mathcal{H}^\text{A}\otimes\mathcal{H}^\text{B}$ with 
$\mathcal{H}^\text{A}=\otimes ^{n-k}\mathbb{C}^2$ and
$\mathcal{H}^\text{B}=\otimes ^{k}\mathbb{C}^2$. For $k\leq n/2$ we will say that a subspace, $V$,
of $\mathcal{H}$ is $k$-\emph{totally entangled} if it is maximally entangled relative to every decomposition
of $\mathcal{H}$ as above.
\end{definition}

It is interesting to note that all the subspaces spanned by
the logical codewords of the different non-degenerate error correcting codes
given in~\cite{Ste96,Ben96,Laf96}  are 2-totally entangled subspaces.
On the other hand, the subspaces spanned by
the logical codewords of degenerate codes, like Shor's 9 qubits code, are 
in general only partially maximally entangled subspaces (i.e. maximally entangled
for some choices of $k$ qubits but not for all choices).
In the following subsections we will see the reason for that.

\subsection{Analysis of Shor's 9 qubits code}

We start with the following notations.
We set $u_{\pm}=\frac{1}{\sqrt{2}}(\left\vert 000\right\rangle \pm\left\vert
111\right\rangle )$ so that the two logical codewords in Shor's 9 qubit code are
$v_{+}=u_{+}\otimes u_{+}\otimes u_{+}$ and $v_{-}=u_{-}\otimes u_{-}\otimes
u_{-}$. The subspace spanned by these codewords is denoted by
$V_{0}=\mathbb{C}v_{+} \oplus \mathbb{C}v_{-}$. We also denote
$u_{\pm}^0=(|100\rangle\pm|011\rangle)/\sqrt{2}$, $u_{\pm}^1=(|010\rangle\pm|101\rangle)/\sqrt{2}$
and $u_{\pm}^2=(|001\rangle\pm|110\rangle)/\sqrt{2}$.

Using these notations, we define $21$ mutually orthogonal $2$ dimensional subspaces
orthogonal to $V_{0}$: 

$V_{1}=\mathbb{C}u_{-}\otimes u_{+}\otimes
u_{+}\oplus\mathbb{C}u_{+}\otimes u_{-}\otimes u_{-}$,

$V_{2}= \mathbb{C}u_{+}\otimes u_{-}\otimes
u_{+}\oplus\mathbb{C}u_{-}\otimes u_{+}\otimes u_{-}$,

$V_{3}=\mathbb{C}u_{+}\otimes u_{+}\otimes
u_{-}\oplus\mathbb{C}u_{-}\otimes u_{-}\otimes u_{+}$,

$V_{4+i}=\mathbb{C}u_{+}^{i}\otimes u_{+}\otimes
u_{+}\oplus\mathbb{C} u_{-}^{i}\otimes u_{-}\otimes u_{-}$, for
$i=0,1,2$,

$V_{7+i}=\mathbb{C}u_{+}\otimes u_{+}^{i}\otimes u_{+}\oplus
\mathbb{C} u_{-}\otimes u_{-}^{i}\otimes u_{-}$, for $i=0,1,2$,

$V_{10+i}=\mathbb{C}u_{+}\otimes u_{+}\otimes u_{+}^{i}\oplus
\mathbb{C}u_{-}\otimes u_{-}\otimes u_{-}^{i}$, for $i=0,1,2$,

$V_{13+i}=\mathbb{C}u_{-}^{i}\otimes u_{+}\otimes
u_{+}\oplus\mathbb{C}u_{+}^{i}\otimes u_{-}\otimes u_{-}$, for
$i=0,1,2$,

$V_{16+i}=\mathbb{C}u_{+}\otimes u_{-}^{i}\otimes
u_{+}\oplus\mathbb{C} u_{-}\otimes u_{+}^{i}\otimes u_{-}$, for
$i=0,1,2$,

$V_{19+i}=\mathbb{C}u_{+}\otimes u_{+}\otimes
u_{-}^{i}\oplus\mathbb{C} u_{-}\otimes u_{-}\otimes u_{+}^{i}$, for
$i=0,1,2$.

\bigskip
If
$X\in End(\mathbb{C}^{2})$ (linear maps of $ \mathbb{C}^{2}$ to
$\mathbb{C}^{2}$) then we denote by $X_{i}$ the linear map of
$\otimes^{9}\mathbb{C}^{2}$ to itself that is the tensor product of
the identity of $\mathbb{C}^{2}$ in every tensor factor but the
$i$-th and is $X$ in the $i$-th factor
thus
\[X_{0}=X\otimes I\otimes I\otimes I\otimes I\otimes I\otimes I\otimes I\otimes
I,\]
\[X_{1}=I\otimes X\otimes I\otimes I\otimes I\otimes I\otimes I\otimes I\otimes I\] etc.

Then we have (here $\left\lfloor x\right\rfloor =\max\{m|m\leq
x,m\in\mathbb{Z}\})$

\[X_{i}V_{0}\subset V_{0}\oplus V_{\left\lfloor i/3\right\rfloor+1 }\oplus
V_{i+4}\oplus V_{i+13},0\leq i\leq8.\] We choose an observable $R$
with
\[R_{|V_{i}}=\lambda_{i}I,0\leq i\leq21\]
and
\[R_{|W}=\mu I,\]
where $W$ is the orthogonal complement of $\oplus_{i=0}^{21}V_{i}$ and $\lambda_{i}\neq\lambda_{j}$ for $i\neq j$ and
$\lambda_{i}\neq\mu$ for any $i$. We define a unitary operator
$U_{j}:V_{j}\rightarrow V_{0} $ as follows:
we denote the Pauli matrices by
\[A_{1}=\left[
\begin{array}
[c]{cc}
1 & 0\\
0 & -1
\end{array}
\right]  ,A_{2}=\left[
\begin{array}
[c]{cc}
0 & 1\\
1 & 0
\end{array}
\right]  ,A_{3}=\left[
\begin{array}
[c]{cc}
0 & i\\
-i & 0
\end{array} 
\right]  ,
\]
and then define $U_{0}=I$, $U_{i}=(A_{1})_{3i-1}$, for $i=1,2,3$, $U_{i}=(A_{2})_{i-4}$, 
for $4\leq i\leq12$ and $U_{i}=(A_{3})_{i-13}$, for $13\leq i\leq21$. This gives an one qubit
error correcting code since if $v\in V_{0}$ is a state and if we have an error
in the $i$-th position then we will have
\[
X_{i}v\in V_{0}\oplus V_{\left\lfloor i/3 \right\rfloor+1 }\oplus
V_{i+4}\oplus V_{i+13}\;.
\]
Thus, if we measure the observable $R$ on $X_{i}v$ then the measurement will
yield one of $\lambda_{j}$ with $j=0,\left\lfloor i/3\right\rfloor+1
,i+4$ or $i+13$ and $X_{i}v$ will have collapsed up to a phase to $U_{j}v$; hence
applying $U_{j}$ will fix the error.

\begin{remark}
Note that the subspace $V_0$ is \emph{not} 2-totally entangled subspace. Nevertheless, $V_0$
has very special properties. In particular, if we group the 9 qubits as $(1,2,3):(4,5,6):(7,8,9)$,
then for any choice of 2 qubits that are not from the same group, the subspace $V_0$ is maximally
entangled with respect to the decomposition between the 2 qubits and the rest 7 qubits. If the 2 qubits
are chosen from the same group then the entanglement of $V_0$ with respect to this decomposition
is $1ebit$. Thus, out of the 36 different decompositions, with respect to 27 of them $\mathcal{E}(V_0)=2ebits$
and with respect to the other 9 decompositions $\mathcal{E}(V_0)=1ebit$.
\end{remark}

\subsection{Orthogonal codes}

We now consider a somewhat more intuitive class of codes known as 
non-degenerate codes which we also name as orthogonal codes.

\begin{definition}\label{orth}
Let
$A_{0}=I,A_{1},A_{2},A_{3}$ the Pauli basis and define $A_{i_{0}i_{1}\cdots
i_{k-1}}^{j_{0}j_{1}\cdots j_{k-1}}$ to be
\[
(A_{j_{0}}\otimes A_{j_{1}}\otimes\cdots\otimes A_{j_{k-1}})_{i_{0}\cdots
i_{k-1}},
\]
where $0\leq j_{r}\leq3$ and $0\leq i_{0}<i_{1}<...<i_{k-1}\leq n-1$. Let $\Sigma$
be the set of distinct operators of the form $A_{i_{0}i_{1}\cdots i_{k-1}
}^{j_{0}j_{1}\cdots j_{k-1}}$. Then an orthogonal $(n,l,k)$ code is an $(n,l,k)$ error correcting code
such that if we label $\Sigma$ as the set of $d+1$ operators $S_{0}=I,S_{1},...,S_{d}$ then $V_{j}
=S_{j}V_{0}$. 
\end{definition}

Note that $\Sigma$ has
\[
d+1=\sum_{r=0}^{k}3^{r}\binom{n}{r}
\]
elements. Thus, a necessary condition that there exist an $(n,l,k)$ code
is the quantum Hamming bound~\cite{NC}:
\[
\sum_{r=0}^{k}3^{r}\binom{n}{r}\leq2^{n-l}\text{.}
\]

\begin{proposition}
A $2^{l}$ dimensional subspace $V$ of $\otimes^{n}\mathbb{C} ^{2}$
is the $V_{0}$ of an $(n,l,k)$-orthoganal error correcting code if and only if
$V$ is $2k$-totally entangled. 
\end{proposition}

\begin{proof}
Let $V$ be a $2k$-totally entangled subspace in $\mathcal{H}=\otimes ^{n}\mathbb{C}^2$,
and let $X: \otimes ^{k}\mathbb{C}^2\rightarrow\otimes ^k\mathbb{C}^2$ be a linear map
on $k$ qubits. As above, for any $i_0<i_1<...<i_{k-1}$ ($1\leq i_l\leq n$) we denote by 
$X_{i_0i_1...i_{k-1}}$ the operation $X$ on $\mathcal{H}$, when acting on the $k$ qubits 
$i_0,i_1,....,i_{k-1}$ (the rest of the $n-k$ qubits are left "untouched"). 
Let also $\mathcal{Z}\equiv\{X\in\text{End}\left(\otimes ^k\mathbb{C}^2\right)|\text{Tr}X=0\}$
and for any $i_0<...<i_{k-1}$ let $U_{i_0...i_{k-1}}\equiv\{X_{i_0...i_{k-1}}V|X\in\mathcal{Z}\}$.
We define the subspace 
$$
\mathcal{W}=V+\sum_{i_0<...<i_{k-1}}U_{i_0...i_{k-1}}\;.
$$
That is, $\mathcal{W}$ consists of all the possible states after an error on $k$ or less
qubits has been occurred.
Now, let $A_0=I,\;A_1,\;A_2,\;A_3$ be an orthonormal basis
of $\text{End}(\mathbb{C}^2)$ with $A_i$ invertible (e.g. the Pauli basis of $2\times 2$
matrices). As in Definition~\ref{orth}, we denote by $A_{i_0...i_{k-1}}^{j_0...j_{k-1}}$ the operator 
$X_{i_0...i_{k-1}}$ that corresponds to $X=A_{j_0}\otimes\cdots\otimes A_{j_{k-1}}$,
and the set of all such operators we denote by 
$$
\Sigma \equiv \{A_{i_0...i_{k-1}}^{j_0...j_{k-1}}|1\leq i_l\leq n,\;0\leq j_l\leq 3\}\;.
$$
Now, let $v_1,...,v_d$ be an orthonormal basis of $V$ and define
$$
\mathfrak{B}\equiv \{Sv_i\Big|S\in\Sigma,\;1\leq i\leq d\}.
$$
We now argue that $\mathfrak{B}$ is an orthonormal basis of $\mathcal{W}$.
Clearly, the vectors in $\mathfrak{B}$ span $\mathcal{W}$. It is therefore enough
to show that the vectors in $\mathfrak{B}$ are orthogonal. Let $Sv_i$ and $S'v_{j}$
be two vectors in $\mathfrak{B}$ with $S=A_{i_0...i_{k-1}}^{j_0...j_{k-1}}$ and
$S'=A_{i_0'...i_{k-1}'}^{j_0'...j_{k-1}'}$. We denote by $\mathcal{H}^\text{B}$ the 
Hilbert space of the qubits $i_0,...,i_{k-1}$ and $i_0',...,i_{k-1}'$,
and by $\mathcal{H}^\text{A}$ the Hilbert space of the rest of the qubits. Note
that $\mathcal{H}^\text{B}$ consists of at most $2k$ qubits. Now, since $V$ is
$2k$-totally entangled subspace, it is maximally entangled relative to the decomposition
$\mathcal{H}=\mathcal{H}^\text{A}\otimes\mathcal{H}^\text{B}$. Thus, from 
Proposition~\ref{isometry} we clearly have
$$
\langle Sv_i|S'v_j\rangle
={\langle v_i|v_j\rangle\text{Tr}\left(\tilde{S}^{\dag}\tilde{S}'\right)}=\delta_{ij}\delta_{SS'}\;,
$$
where $S\equiv I^\text{A}\otimes\tilde{S}$ and $S'\equiv I^\text{A}\otimes\tilde{S}'$; that is,
$\tilde{S}$ and $\tilde{S}'$ are the projections of $S$ and $S'$ onto 
$\mathcal{H}^\text{B}$, respectively. Hence, $\mathfrak{B}$ is an orthonormal basis
of $\mathcal{W}$.

Since $\mathfrak{B}$ is an orthonormal basis we can construct an observable
(i.e. Hermition operator) $R$ such that for all $v\in V$ $R(Sv)=\lambda_{S}Sv$ with all of the 
$\lambda_{S}$ distinct. We also define $R$ to be zero on the orthogonal complement
to $\mathcal{W}$ in $\mathcal{H}$. Now, suppose that an element $v$ has been changed
by a $k$-qubit transformation yielding $X_{i_0...i_{k-1}}v$. We do a measurment of $R$
and since the image is in $\mathcal{W}$ the outcome is $\lambda_{S}$ for some $S$.
After the measurment, the quantum state is $Sv$ and so we recover $v$ by applying
$S^{-1}$ (actually $S$ if we used the Pauli basis).  The converse follows from the same
lines in the opposite direction. This completes the proof.
\end{proof}

Note that Corollary~\ref{cordim} together with the 
proposition above is consistent with the quantum Singleton 
bound~\cite{Kni97}, $n\geq 4k+l$, which also follows
trivially from the quantum Hamming bound for the case 
of orthogonal codes that we considered in this subsection.

\section{Summery and Conclusions}

We introduced the notion of entanglement of subspaces as a measure that quantify the entanglement of
bipartite states in a randomly selected subspace. We discussed its properties and suggested that
it is  additive. We were not able to prove this conjecture (which is equivalent to the additivity
conjecture of the entanglement of formation) although some numerical tests~\cite{HarXX} supports
that and for maximally entangled subspaces we proved that it is additive. We then extended the definition of
maximally entangled subspaces into $k$-totally entangled subspaces and showed that the later
can play an important role in the study of quantum error correction codes.

We considered both degenerate and non-degenerate codes and showed that the subspace spanned
by the logical codewords of a non-degenerate code is a $k$-totally (maximally) entangled
subspace. This observation, followed by an analysis of the degenerate Shor's nine qubits code in terms of 22 mutually 
orthogonal subspaces, motivated us to define a general (possible degenerate) error correcting code in terms of
subspaces. We believe that further investigation in this direction would lead to a better understanding of degenerate
quantum error correcting codes.

\emph{Acknowledgments:---}
The authors would like to thank Aram Harrow and David Meyer for fruitful
discussions.

\end{document}